\documentclass[letterpaper, 10 pt, conference]{ieeeconf}  

\IEEEoverridecommandlockouts                              
\usepackage{graphics} 
\usepackage{booktabs}
\usepackage{amsmath}
\usepackage{amssymb}
\usepackage{graphicx}
\usepackage{algorithm}
\usepackage{algorithmic}
\usepackage{bm}
\usepackage{mathtools}
\usepackage{tikz}
\usepackage{enumerate}
\usetikzlibrary{arrows.meta,positioning,fit,calc,backgrounds,shadows}
\usepackage{stfloats}

\newtheorem{theorem}{Theorem}

\newtheorem{remark}{Remark}

\newcommand{\R}{\mathbb{R}}

\newcommand{\Prob}{\mathbb{P}}

\newcommand{\cA}{\mathcal{A}}

\newcommand{\DeltaTheta}{\Delta(\Theta)}

\title{\LARGE \bf
Fast Strategy Solving for the Informed Player in Two-Player Zero-Sum Linear-Quadratic Differential Games with One-Sided Information}

\author{Mukesh Ghimire, Zhe Xu, Yi Ren
\thanks{M. Ghimire, Z. Xu, and Y. Ren are with the Department of Mechanical and Aerospace Engineering at
        Arizona State University, Tempe, AZ 85281, USA.
        {\tt\small \{mghimire, xzhe1, yiren\}@asu.edu}}%
}

\begin{document}

\maketitle
\thispagestyle{empty}
\pagestyle{empty}

\begin{abstract}
We study finite-horizon two-player zero-sum differential games with one-sided payoff information ($G$), where the informed player (P1) knows the game payoff, while P2 only has a public belief over a finite set of possible payoffs.
In this case, P1's Nash equilibrium (NE) behavioral strategy may control the release of the type information or even resort to manipulate P2's belief. 
Previous studies revealed an atomic structure of the NE of $G$ with general nonlinear dynamics and payoffs, leading to tractable NE approximation.
Implementing such approximation schemes for real-time sub-game solving, however, has not been achieved, yet is desired for applications where sim-to-real gaps exist and robust control is required.
This paper improves the computational efficiency of sub-game solving for P1 during $G$ with linear dynamics and quadratic losses.
Specifically, we show that P1's NE computation can be formulated as a bi-level optimization problem where the outer level optimizes the ``signaling'' strategy, i.e., when and how to reveal information through control, and the inner level is a game-tree LQR that solves for the optimal closed-loop control. 
This bi-level problem is solved via an adjoint-enabled backpropagation scheme: A ``backward'' LQR pass is followed by a ``forward'' gradient descent pass for improving the signaling.
We apply the proposed algorithm to approximate NEs for variants of a homing problem with a 8D state space, 2D action spaces, and a discrete time horizon of $K=10$. The algorithm achieves $\approx$10Hz sub-game solving, enabling robust game-theoretic planning under information asymmetry and random disturbances.
\end{abstract}


\section{INTRODUCTION} 
Two-player zero-sum differential games with one-sided information ($G$) arise in safety-critical scenarios in defense~\cite{garcia2018design, garcia2021complete, liang2019differential}, cybersecurity~\cite{durkota2017optimal, mc2016data, horak2016point, van2013flipit}, and finance~\cite{kyle1985continuous, back1992insider, wang2020market}. When only one player observes the payoff type, optimal play involves both control and \emph{information design}: the informed player may act cautiously to avoid revealing its type too early, or deliberately reveal to induce advantageous responses.
Existing algorithms for solving $G$~\cite{ghimire2025solvingfootballexploitingequilibrium, ghimirestate} identified the atomic NE structure of $G$, but real-time subgame solving, e.g., in robotics applications where sim-to-real gap exists, is not yet achieved.
To this end, this work focuses on a subclass of $G$ with \emph{linear dynamics, quadratic costs} and where the number of payoff types $I$ is \emph{small} (denoted LQ $G$). 
We investigate the effectiveness of a game solver that combines the two structural insights to enable sub-second strategy computation in LQ $G$: (1) the atomic NE structure reduces the game tree branching factor from the action size (infinite for a continuous action space) to at most $I$ per infostate (tree node), resulting in sub-game solving as a minimax problem on an $I$-ary game tree, and (2) the LQ structure ensures that the minimax problem can be solved via a differentiable Riccati recursion. 

The main contributions are:
\begin{enumerate}
    \item The formulation of P1's NE strategy learning as a bi-level optimization problem where a belief-induced LQ game is nested in signaling optimization.

    \item An adjoint-enabled backpropagation algorithm that differentiates the Riccati recursion to optimize signaling policy, resulting in $\approx$10Hz sub-game solving for a homing game with an 8D state space, 2D action spaces, and a time horizon of $K=10$.

    \item Two case studies of the homing game: we use Hexner's game with an analytical NE to verify the correctness of our algorithm, and a ``drone landing under attack'' scenario to demonstrate the value of fast sub-game solving under P2's random velocity disturbances.
\end{enumerate}

\section{RELATED WORK}
Games with incomplete information have been studied extensively since Harsanyi's foundational work on Bayesian games~\cite{harsanyi1967games}. In the zero-sum setting, Aumann and Maschler~\cite{aumann1995repeated} characterized optimal strategies in repeated games with one-sided information, establishing that there are scenarios when the informed player should reveal information gradually through a carefully designed signaling scheme. 
For differential games, Cardaliaguet~\cite{cardaliaguet2007differential, cardaliaguet2009numerical} established value existence and characterized the value functions for games with asymmetric information using viscosity solution methods. Ghimire et al.~\cite{ghimire2025solvingfootballexploitingequilibrium} built on these ideas to propose a scalable solver for games with one-sided information and nonlinear dynamics and costs. The key insight from these studies is that the optimal strategy for the informed (resp. uninformed) player entails randomizing from at most $I$ (resp. $I+1$) actions at any game tree node, enabling efficient strategy computation for large-scale games such as football. 
Hexner~\cite{hexner1979differential} independently studied LQ $G$ through a homing game, where an evader with a private target attempts to minimize terminal distance, while a pursuer infers the target and maximizes it. With an LQ structure, Hexner derived the analytical solution to optimal signaling and feedback control. 

\section{PROBLEM FORMULATION}
\label{sec:formulation}

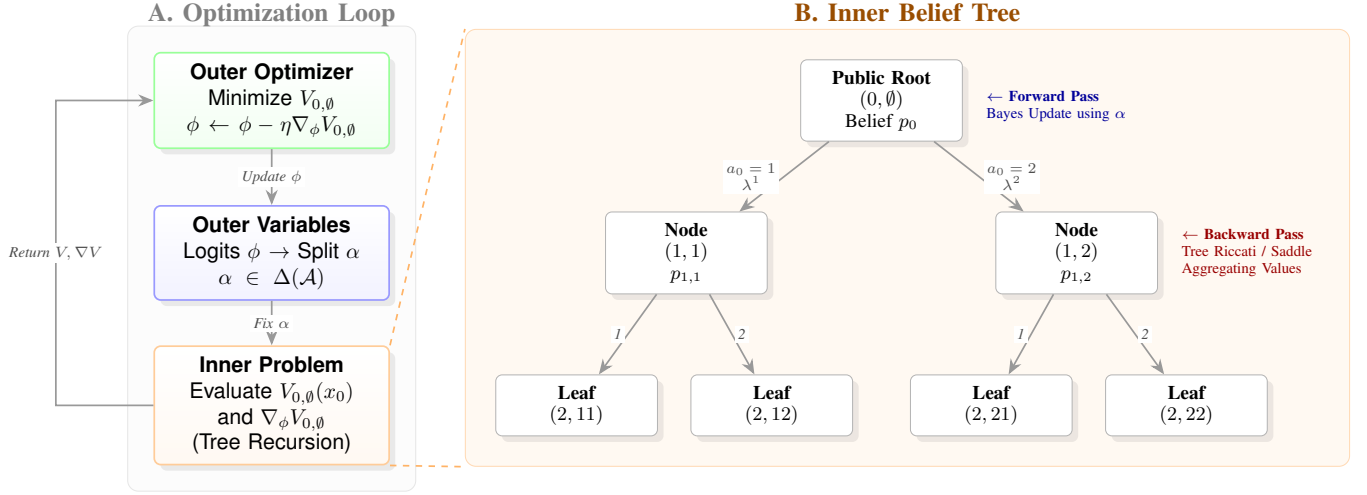
\begin{figure*}[htbp]
\vspace{0.05in}
\centering
\resizebox{\textwidth}{!}{%
\begin{tikzpicture}[
    font=\sffamily,
    >=Stealth,
    box/.style={
        rectangle,
        draw=gray!40,
        top color=white,
        bottom color=gray!5,
        rounded corners=3pt,
        align=center,
        inner sep=5pt,
        drop shadow={opacity=0.15, shadow xshift=1mm, shadow yshift=-1mm}
    },
    optbox/.style={box, bottom color=green!5, draw=green!40, thick, text width=35mm, minimum height=12mm},
    varbox/.style={box, bottom color=blue!5, draw=blue!40, thick, text width=35mm, minimum height=12mm},
    procbox/.style={box, bottom color=orange!5, draw=orange!40, thick, text width=35mm, minimum height=12mm},
    treenode/.style={box, bottom color=white, draw=gray!60, text width=23mm, minimum height=8mm, font=\small},
    line/.style={draw, thick, -{Stealth[length=2.5mm, width=2mm]}, gray!80},
    dashline/.style={draw, thick, dashed, -{Stealth[length=2.5mm, width=2mm]}, gray!80},
    zoomline/.style={draw, dashed, color=orange!60, thick},
    lbl/.style={midway, align=center, font=\scriptsize\itshape, text=black!70, fill=white, inner sep=1.5pt}
]

    
    \node[optbox] (opt) at (0, 0) {
        \textbf{Outer Optimizer}\\
        Minimize $V_{0,\emptyset}$\\
        $\phi \leftarrow \phi - \eta\nabla_\phi V_{0, \emptyset}$
    };
    
    \node[varbox] (alpha) at (0, -2.5) {
        \textbf{Outer Variables}\\
        Logits $\phi \to$ Split $\alpha$\\
        $\alpha \in \Delta(\cA)$
    };
    
    \node[procbox] (inner_macro) at (0, -5) {
        \textbf{Inner Problem}\\
        Evaluate $V_{0, \emptyset}(x_0)$ and $\nabla_\phi V_{0, \emptyset}$\\
        (Tree Recursion)
    };
    
    \draw[line] (opt) -- node[lbl] {Update $\phi$} (alpha);
    \draw[line] (alpha) -- node[lbl] {Fix $\alpha$} (inner_macro);
    
    \draw[line] (inner_macro.west) -- ++(-1.6,0) |- node[lbl, pos=0.25] {Return $V, \nabla V$} (opt.west);

    \begin{pgfonlayer}{background}
        \node[draw=gray!20, fill=gray!2, rounded corners=5pt, fit=(opt)(inner_macro)(alpha), 
              inner sep=12pt, label={[font=\bfseries\large, gray, yshift=-4pt]above:A. Optimization Loop}] (macro_bg) {};
    \end{pgfonlayer}

    
    \node[treenode] (root) at (10, 0) {
        \textbf{Public Root}\\
        $(0,\emptyset)$\\
        Belief $p_{0}$
    };

    \node[treenode] (n1) at (6.8, -2.5) {
        \textbf{Node}\\
        $(1,1)$\\
        $p_{1,1}$
    };
    \node[treenode] (n2) at (13.2, -2.5) {
        \textbf{Node}\\
        $(1,2)$\\
        $p_{1,2}$
    };

    \node[treenode] (n11) at (5.0, -5) {
        \textbf{Leaf}\\
        $(2,11)$
    };
    \node[treenode] (n12) at (8.2, -5) { 
        \textbf{Leaf}\\
        $(2,12)$
    };
    \node[treenode] (n21) at (11.8, -5) {
        \textbf{Leaf}\\
        $(2,21)$
    };
    \node[treenode] (n22) at (15.0, -5) {
        \textbf{Leaf}\\
        $(2,22)$
    };

    \draw[line] (root) -- node[lbl, left=2pt] {$a_0=1$\\$\lambda^1$} (n1);
    \draw[line] (root) -- node[lbl, right=2pt] {$a_0=2$\\$\lambda^2$} (n2);
    
    \draw[line] (n1) -- node[lbl, left=1pt] {1} (n11);
    \draw[line] (n1) -- node[lbl, right=1pt] {2} (n12);
    \draw[line] (n2) -- node[lbl, left=1pt] {1} (n21);
    \draw[line] (n2) -- node[lbl, right=1pt] {2} (n22);

    \node[align=left, font=\scriptsize, color=blue!60!black, anchor=west] (bayes_txt) at (11.6, -0.2) {
        $\leftarrow$ \textbf{Forward Pass}\\
        Bayes Update using $\alpha$\\
    };
    
    \node[align=left, font=\scriptsize, color=red!60!black, anchor=west] (riccati_txt) at (14.8, -2.5) {
        $\leftarrow$ \textbf{Backward Pass}\\
        Tree Riccati / Saddle\\
        Aggregating Values
    };
    
    \begin{pgfonlayer}{background}
         \node[draw=orange!30, fill=orange!5, rounded corners=5pt, fit=(root)(n11)(n22)(bayes_txt)(riccati_txt), 
         inner sep=14pt, label={[font=\bfseries\large, orange!60!black]above:B. Inner Belief Tree}] (tree_bg) {};
    \end{pgfonlayer}

    
    \draw[zoomline] (inner_macro.north east) -- (tree_bg.north west);
    \draw[zoomline] (inner_macro.south east) -- (tree_bg.south west);

\end{tikzpicture}%
}
\caption{Bilevel structure used in the solver. (A) The outer level optimizes belief-splitting coefficients $\alpha$ (parameterized by logits $\phi$). (B) Given $\alpha$, the inner level constructs the public belief tree (forward Bayes propagation), computes belief-averaged costs, and evaluates the induced tree of LQ saddle subgames via a tree-structured Riccati recursion, returning the root value and its gradient.}
\label{fig:consolidated_bilevel}
\end{figure*}

\subsection{Game Setup and Information}
We consider LQ $G$ on a finite time horizon $[0,T]$, and approximate its NE through time discretization with a time step
$\tau>0$. Let $K\coloneqq T/\tau$ and $t_k\coloneqq k\tau$ for $k\in[K]$.
At $t=0$, a payoff type $\theta$ is drawn from
$\Theta=\{\theta_i\}_{i\in[I]}$ with a public prior $p_0\in\DeltaTheta$.
P1 observes $\theta$; P2 only knows $p_0$.
In continuous time, the joint state $x(t)\in\R^n$ evolves according to linear dynamics
\[
\dot{x}(t)=A_cx(t)+B_{1,c}u(t)+B_{2,c}v(t),
\]
where $u(t)\in\R^{m_1}$ is P1's control and $v(t)\in\R^{m_2}$ is P2's control.
The corresponding discrete-time dynamics is:
\begin{equation}
x_{k+1}=A x_k + B_1 u_k + B_2 v_k,
\label{eq:dynamics}
\end{equation}
with $A=I+\tau A_c$, $B_i=\tau B_{i,c} + A_cB_{i,c}\frac{\tau^2}{2}$ for $i\in[2]$. 
The public belief about $\theta$ is denoted by $p(t)$ in continuous time and $p_k$ in discrete time, and is updated by Bayes' rule.
Throughout the game both players observe the state, action, and belief trajectories.

\subsection{Strategy and Payoff Structures}
We denote by $\{\mathcal{H}_r^i\}^I$ the joint sets of $I$ behavioral strategies of P1\footnote{We define ``behavioral strategy'' as the time-delayed non-anticipative strategies standard for differential games~\cite{elliott1972existence}.}.
A type-$\theta_i$ behavioral strategy $\eta_i \in \mathcal{H}_r^i$ maps the current public information
(e.g., $(t,x,p)$ in continuous time or $(k,x_k,p_k)$ in discrete time) to a probability measure over $\R^{m_1}$.
Similarly, P2's behavioral strategy $\zeta \in \mathcal{Z}_r$ produces a probability measure over $\R^{m_2}$.
The subscript $r$ highlights the randomness of behavioral strategies.

Let $X_T$ be the random terminal state induced by $(\eta_i,\zeta)$ under the continuous-time dynamics.
Type-$\theta_i$ P1 minimizes
\begin{equation}\label{eq:cont_cost}
    J_{\theta_i}(t_0, x_0) := \mathbb{E}_{\eta_i, \zeta}\left[g_{\theta_i}(X_T) + \int_{t_0}^{T}\ell_{\theta_i}(\eta_i(s), \zeta(s))\;ds\right],
\end{equation}
where the running cost is quadratic and strictly convex-concave in controls
\begin{equation}
\ell_{\theta_i}(u,v)=\frac12 u^\top R_i u \;-\;\frac12 v^\top S_i v,
\quad
R_i\succ 0,\;S_i\succ 0,
\label{eq:stage_cost}
\end{equation}
and the terminal cost is convex in P1's state and concave in P2's:
\begin{equation}
g_{\theta_i}(x)=\frac12 x^\top Q_i x + q_i^\top x + c_i,
\;\; Q_i=Q_i^\top,\;q_i\in\R^n,\;c_i\in\R.
\label{eq:terminal_cost_general}
\end{equation}

P2 maximizes the expected payoff $J(t_0,x_0,p_0)=\mathbb{E}_{\theta\sim p_0}[J_{\theta}(t_0,x_0)]$.
The upper and lower values are defined as $V^+(t_0,x_0,p_0)=\inf_{\{\eta_i\}}\sup_{\zeta}J$ and
$V^-(t_0,x_0,p_0)=\sup_{\zeta}\inf_{\{\eta_i\}}J$, respectively.
The game has a value $V$ if $V = V^+(t_0,x_0,p_0) = V^-(t_0,x_0,p_0)$.
Existence of $V$ in $G$ is proved in \cite{cardaliaguet2009numerical}, notably under the Isaacs' condition which holds for LQ $G$.
\cite{ghimire2025solvingfootballexploitingequilibrium} develops a discrete-time primal-dual reformulation of the game to approximate NE strategies.
Specifically, P1 solves a Stackelberg primal game where P2 best responses to P1 via the following Bellman backup:
\begin{equation}\label{eq:primal}
\begin{aligned}
    V_\tau(k, x_k, p_k) &= \min_{\{\eta_i\} \in \{\mathcal H_r\}^I}\mathbb{E}_{\theta_i \sim p_k, u_k \sim \eta_i}\Big[\max_{v_k}V_\tau(k+1,\\
    &\hspace{1in} x_{k+1}, p_{k+1}) + \tau \ell_{\theta_i}(u_k, v_k)\Big],
\end{aligned}
\end{equation}
with boundary $V_\tau(K, x, p) = \sum_{\theta_i}p_{\theta_i}g_{\theta_i}(x)$. Here $p_{k+1}$ is obtained after observing action $u$: $p_{k+1}(i)\propto p_k(i)\,\Prob(u\mid \theta=\theta_i)$ under P1's strategy set $\{\eta_i\}$. 

The rest of the paper focuses on efficiently solving \eqref{eq:primal} at all game tree nodes.
Specifically, we approximate $\{\eta_i\}$ by separating (i) a signaling component that governs the evolution of the public belief (i.e., how $p_k$ is split into posteriors $p_{k+1}$), and (ii) a control component that specifies the corresponding optimal state-feedback actions given the public belief. The next subsection makes this separation explicit via an $I$-branch public belief-tree representation.


\subsection{Reparameterization of $\{\eta_i\}$}
Our goal is to decompose $\{\eta_i\}$ into variables for signaling and those for control.
We start with a critical property of $G$: It is proved that P1's NE strategies in \eqref{eq:primal} are \emph{atomic}~\cite{ghimire2025solvingfootballexploitingequilibrium}: $\eta_i$ for $i\in[I]$ randomizes over at most $I$ common action prototypes at each tree node, and 
thus the primal game can be considered a tree-structured minimax problem, with each node consisting of $I$ branches. See visualization in Fig.~\ref{fig:consolidated_bilevel}B for $I=2$ and $K=2$. 

Let $\cA \coloneqq \{1,2,\dots,I\}$ be the index set for action prototypes, and $\omega\in\cA^k$ denote P1's length-$k$ action sequence. Each tree node can be uniquely identified by $(k,\omega)$.
At each node $(k,\omega)$ and for each private type $\theta_i$, P1 selects the next action index $a\in\cA$ according to a distribution specified by $\eta_i$:
\begin{equation}
\begin{aligned}
&\alpha_{k,\omega,i}^a \;=\;\Prob(a_k=a \mid \theta=\theta_i,\; \text{public node }(k,\omega)),\\
&\alpha_{k,\omega,i}^a\ge 0,\;\;
\sum_{a\in\cA}\alpha_{k,\omega,i}^a=1.
\end{aligned}
\label{eq:alpha_def}
\end{equation}
We call the collection $\alpha=\{\alpha_{k,\omega,i}^a\}$ the signaling policy as it governs how the public belief evolves: 
Given a public belief $p\in\DeltaTheta$ at node $(k,\omega)$, the probability of observing prototype $a$ is
\begin{equation}
\lambda_{k,\omega}^a(p)\coloneqq \sum_{i=1}^I p_i\,\alpha_{k,\omega,i}^a.
\label{eq:lambda}
\end{equation}
After observing $a$, Bayes' rule yields the posterior belief at child node $(k\!+\!1,\omega a)$:
\begin{equation}
p_{k+1,\omega a}(i)=\frac{\alpha_{k,\omega,i}^a\,p_{k,\omega}(i)}{\lambda_{k,\omega}^a(p_{k,\omega})}
\qquad
(\lambda_{k,\omega}^a(p_{k,\omega})>0).
\label{eq:bayes}
\end{equation}
As a result, belief $p$ across the game tree follows a martingale determined by $\alpha$: 
$\sum_{a\in\cA}\lambda_{k,\omega}^a(p_{k,\omega})\,p_{k+1,\omega a}=p_{k,\omega}$.

Now we describe the LQ game for fixed $\alpha$. Denote the value at $(k,\omega)$ as $V_{k,\omega}(x)$.
At leaves, we have:
\begin{equation}
V_{K,\omega}(x)=\sum_{i=1}^I p_{K,\omega}(i)\,g_i(x).
\label{eq:terminal_value}
\end{equation}
Consider an edge $(k,\omega,a)$ from node $(k,\omega)$ to child $(k\!+\!1,\omega a)$.
Conditioned on observing $a$, the posterior belief is $p_{k+1,\omega a}$ by \eqref{eq:bayes}.
Taking expectation of the type-dependent running cost \eqref{eq:stage_cost} under this posterior yields
a standard quadratic stage cost with belief-averaged matrices
\begin{equation}
\bar R_{k,\omega}^a \;\coloneqq\; \sum_{i=1}^I p_{k+1,\omega a}(i)\,R_i,
\quad
\bar S_{k,\omega}^a \;\coloneqq\; \sum_{i=1}^I p_{k+1,\omega a}(i)\,S_i.
\label{eq:running_avg}
\end{equation}
Let $\lambda_{k,\omega}^a := \lambda_{k,\omega}^a(p_{k,\omega})$ for brevity.
The edge value conditioning on $a$ solves an LQ subgame:
\begin{equation}
\begin{aligned}
V_{k,\omega}^a(x)
\;\coloneqq\;
\min_{u\in\R^{m_1}}&\max_{v\in\R^{m_2}}
\Big\{
\tau\Big(\tfrac12 u^\top \bar R_{k,\omega}^a\,u \;-\;\tfrac12 v^\top \bar S_{k,\omega}^a\,v\Big)\\
&+\,V_{k+1,\omega a}\!\big(Ax+B_1u+B_2v\big)
\Big\}.
\end{aligned}
\label{eq:edge_value_def}
\end{equation}
Since Isaacs' condition holds for LQ $G$, 
each edge subproblem \eqref{eq:edge_value_def} has a unique solution. This will be formalized in Theorem~\ref{prop:quad_closure_tree}.
The node value is an expectation over branches:
\begin{equation}
V_{k,\omega}(x)
=\sum_{a\in\cA}\lambda_{k,\omega}^a\,V_{k,\omega}^a(x),
\quad k=0,\dots,K-1.
\label{eq:tree_bellman}
\end{equation}

To summarize, the equilibrium computation of \eqref{eq:primal} across the game tree can be viewed as a bilevel program:
the outer level chooses the signaling $\alpha$, while the inner
level solves the induced  game~\eqref{eq:terminal_value}--\eqref{eq:tree_bellman} to define an evaluation map
$\alpha \mapsto V_{0,\emptyset}(x_0, p_0)$.
With initial state $x_0$ and prior $p_0$, the outer signaling problem is
\begin{equation}
\min_{\alpha}\; V_{0,\emptyset}(x_0, p_0)
\quad
\text{s.t.}\quad
\alpha_{k,\omega,i}\in\Delta(\cA)\;\;\forall k,\omega,i.
\label{eq:outer_problem}
\end{equation}
In Section~\ref{sec:riccati} we show how to (i) evaluate $V_{0,\emptyset}(x_0,p_0)$
efficiently for a fixed $\alpha$ via a tree-structured Riccati recursion, and (ii) differentiate through
this evaluation to optimize $\alpha$.

\section{DIFFERENTIABLE RICCATI SOLVER}
\label{sec:riccati}

We now describe a differentiable tree-structured Riccati recursion composed of
(1) \emph{edge steps} that solve the local LQ games \eqref{eq:edge_value_def} to obtain feedback laws and edge values $V_{k,\omega}^a$ given the values at child nodes $V_{k+1,\omega a}$, and  
(2) \emph{node steps} \eqref{eq:tree_bellman} that average edge values using branch probabilities $\lambda_{k,\omega}^a$ to obtain the \emph{quadratic} value at parent node.

\subsection{Quadratic Value Representation on Nodes}
\label{subsec:node_quadratic}
A key enabler of the Riccati recursion is Theorem~\ref{prop:quad_closure_tree}, which states that node values are quadratic and convex in state:
\begin{equation}
V_{k,\omega}(x)=\frac12 x^\top P_{k,\omega}x + r_{k,\omega}^\top x + c_{k,\omega},
\label{eq:value_quad}
\end{equation}
where coefficients $(P_{k,\omega},r_{k,\omega},c_{k,\omega})$ depend on $p_{k,\omega}$ (hence on $\alpha$) and can be computed by backward induction.

\begin{theorem}[Quadratic closure on the public tree]
\label{prop:quad_closure_tree}
\textit{For fixed $\alpha$, $V_{k,\omega}(\cdot)$ at any $(k,\omega)$ is quadratic convex in state. }
\end{theorem}

\begin{proof}
We prove by induction.
At a leaf node $(K,\omega)$ \eqref{eq:terminal_value}, since $g_i$ is quadratic and strictly convex in state, $V_{K,\omega}$ is also quadratic and convex. Using \eqref{eq:terminal_cost_general}, leaf coefficients are
\begin{equation}
\begin{aligned}
P_{K,\omega}&=\sum_{i=1}^I p_{K,\omega}(i)\,Q_i,\qquad
r_{K,\omega}=\sum_{i=1}^I p_{K,\omega}(i)\,q_i,\\
c_{K,\omega}&=\sum_{i=1}^I p_{K,\omega}(i)\,c_i.
\end{aligned}
\label{eq:leaf_params}
\end{equation}

For induction, fix an edge $(k,\omega,a)$ with child node $(k\!+\!1,\omega a)$.
Let the child value be
\begin{align*}
&V_{k+1,\omega a}(x)=\frac12 x^\top P^+x + r^{+\top}x + c^+,
\quad
P^+\coloneqq P_{k+1,\omega a}, \\
&\hspace{1in}r^+\coloneqq r_{k+1,\omega a},\quad
c^+\coloneqq c_{k+1,\omega a}.
\end{align*}
On this edge, the posterior belief-averaged running matrices are $\bar R\coloneqq \bar R_{k,\omega}^a$
and $\bar S\coloneqq \bar S_{k,\omega}^a$ from \eqref{eq:running_avg}. The edge value
\eqref{eq:edge_value_def} can therefore be written as
\begin{equation}
\begin{aligned}
V_{k,\omega}^a(x)
=
\min_{u}\max_{v}\;
&\Big\{
\tfrac12 u^\top (\tau\bar R)\,u
-\tfrac12 v^\top (\tau\bar S)\,v \\
& + \tfrac12 x^{+\top}P^+x^+ + r^{+\top}x^+ + c^+
\Big\},
\end{aligned}
\label{eq:edge_problem_explicit}
\end{equation}
where $x^+=Ax+B_1u+B_2v$.

Define the stacked input $w=\begin{bmatrix}u&v\end{bmatrix}^\top$ and the input matrix
$B=[B_1\;\;B_2]$. The saddle stage Hessian is
\begin{equation}
\bar R_w \;\coloneqq\;
\begin{bmatrix}
\tau\bar R & 0\\
0 & -\tau\bar S
\end{bmatrix}.
\label{eq:Rw_def}
\end{equation}
Then the objective inside \eqref{eq:edge_problem_explicit} becomes a quadratic function of $(w,x)$:
\begin{equation}
\begin{aligned}
\Phi_{k,\omega}^a(w;x)
&=
\frac12 w^\top \bar R_w w
+\frac12 (Ax+B w)^\top P^+ (Ax+B w) \\
&\hspace{1in}+ r^{+\top}(Ax+B w) + c^+.
\end{aligned}
\label{eq:Phi_def}
\end{equation}
Expanding \eqref{eq:Phi_def} and collecting terms in $w$ yields
\begin{equation}
\begin{aligned}
\Phi_{k,\omega}^a(w;x)
=
&\;\frac12 w^\top\underbrace{\big(\bar R_w + B^\top P^+ B\big)}_{H_{k,\omega}^a}w\\
&\;+\; w^\top B^\top(P^+A x + r^+)\\
&\;+\;\frac12 x^\top A^\top P^+A x \;+\; (A^\top r^+)^\top x \;+\; c^+.
\end{aligned}
\label{eq:Phi_expanded}
\end{equation}
We first verify that the minimax problem is well-posed and admits a \emph{unique} saddle.
Let $\|\cdot\|$ denote the Euclidean norm for vectors and the induced (spectral) operator norm for matrices, and since $P^+=P^{+\top}$,
$\|P^+\|=\max_i|\lambda_i(P^+)|$, where $\lambda(M)$ is the eigenvalue of $M$.  From \eqref{eq:Phi_expanded},
\[
H_{k,\omega}^a=
\begin{bmatrix}H_{uu}&H_{uv}\\ H_{uv}^\top&H_{vv}\end{bmatrix},
\]
with $H_{uu}=\tau\bar R+B_1^\top P^+B_1,\quad
H_{vv}=-\tau\bar S+B_2^\top P^+B_2$.
Since $\bar R,\bar S$ are posterior convex combinations of $\{R_i\succ0\}$ and $\{S_i\succ0\}$,
$\lambda_{\min}(\bar R)\ge \underline r:=\min_i\lambda_{\min}(R_i)>0$ and
$\lambda_{\min}(\bar S)\ge \underline s:=\min_i\lambda_{\min}(S_i)>0$. Thus, for any $\delta u,\delta v$,
\begin{align}
\delta u^\top H_{uu}\delta u
&\ge \big(\tau\underline r-\|P^+\|\,\|B_1\|^2\big)\|\delta u\|^2,
\label{eq:Huu_bd}\\
\delta v^\top H_{vv}\delta v
&\le \big(-\tau\underline s+\|P^+\|\,\|B_2\|^2\big)\|\delta v\|^2,
\label{eq:Hvv_bd}
\end{align}
with $\delta u^\top\bar R\delta u\ge \underline r\|\delta u\|^2$, $y^\top P^+y\ge-\|P^+\|\|y\|^2$, and
$\|B_j\delta\|\le\|B_j\|\,\|\delta\|$.
Let $\bar P:=\sup_{k,\omega}\|P_{k,\omega}\|<\infty$ and assume a standard discretization so that
$\|B_j(\tau)\|\le \beta_j\tau$ for all $\tau\in(0,\tau_0]$ and some $\beta_j,\tau_0>0$. Then
$H_{uu}\succ0$ and $H_{vv}\prec0$ for any $\tau\in(0,\tau^*)$, where
\[
\tau^*:=\min\left\{\frac{\underline r}{\bar P\beta_1^2},\ \frac{\underline s}{\bar P\beta_2^2},\ \tau_0\right\}.
\]
Hence $\Phi_{k,\omega}^a(\cdot;x)$ is strictly convex in $u$ and strictly concave in $v$, and the edge game
\eqref{eq:edge_problem_explicit} admits a unique saddle point.

Therefore, we have
\[
\nabla_w \Phi_{k,\omega}^a(w;x)
=
H_{k,\omega}^a\,w + B^\top(P^+A x + r^+) = 0.
\]
The unique saddle point is then the following:
\begin{equation}
\begin{aligned}
&w_{k,\omega}^{a*}(x)=\mathcal K_{k,\omega}^a x + \kappa_{k,\omega}^a,
\;\;
\mathcal K_{k,\omega}^a = -\big(H_{k,\omega}^a\big)^{-1}B^\top P^+ A,
\quad \\
&\kappa_{k,\omega}^a = -\big(H_{k,\omega}^a\big)^{-1}B^\top r^+.
\label{eq:edge_feedback}
\end{aligned}
\end{equation}
Partitioning $w=\begin{bmatrix}u&v\end{bmatrix}^\top$ gives the edge-specific affine feedback laws
\begin{align}
u_{k,\omega}^{a*}(x)&=\mathcal K_{u,k,\omega}^a x+\kappa_{u,k,\omega}^a,
\label{eq:op_u}\\
v_{k,\omega}^{a*}(x)&=\mathcal K_{v,k,\omega}^a x+\kappa_{v,k,\omega}^a,
\label{eq:op_v}
\end{align}
where $(\mathcal K_u,\kappa_u)$ and $(\mathcal K_v,\kappa_v)$ are the blocks of $(\mathcal K,\kappa)$.
Let $b(x)\;\coloneqq\;B^\top(P^+A x + r^+)$ be the affine term in \eqref{eq:Phi_expanded}. Then \eqref{eq:Phi_expanded} can be written as
\begin{align*}
\Phi_{k,\omega}^a(w;x)
&=
\frac12 (w + H^{-1}b(x))^\top H (w + H^{-1}b(x))\\
&-\frac12 b(x)^\top H^{-1}b(x)
+\frac12 x^\top A^\top P^+A x \\
&\hspace{1.13in}+ (A^\top r^+)^\top x + c^+,
\end{align*}
where $H\equiv H_{k,\omega}^a$. At the saddle point $$w=w^*(x)=-H^{-1}b(x),$$ the squared term vanishes and we obtain
\begin{equation}
\begin{aligned}
V_{k,\omega}^a(x)
&=
\frac12 x^\top A^\top P^+A x \;+\; (A^\top r^+)^\top x \;+\; c^+\\
&\hspace{1.08in}-\;\frac12\,b(x)^\top \big(H_{k,\omega}^a\big)^{-1} b(x).
\end{aligned}
\label{eq:edge_value_complete_square}
\end{equation}

Expanding the last term with $b(x)=B^\top(P^+A x + r^+)$ and collecting quadratic, linear, and constant parts in $x$
gives the quadratic coefficients of the edge value:
\begin{align}
P_{k,\omega}^a &=
A^\top P^+A - A^\top P^+B\big(H_{k,\omega}^a\big)^{-1}B^\top P^+A,
\label{eq:edge_P}\\
r_{k,\omega}^a &=
A^\top r^+ - A^\top P^+B\big(H_{k,\omega}^a\big)^{-1}B^\top r^+,
\label{eq:edge_r}\\
c_{k,\omega}^a &=
c^+ - \frac12 r^{+\top}B\big(H_{k,\omega}^a\big)^{-1}B^\top r^+.
\label{eq:edge_c}
\end{align}
Returning to the node recursion \eqref{eq:tree_bellman},
since each $V_{k,\omega}^a(x)$ is quadratic in $x$, the node coefficients are simply the $\lambda$-weighted averages:
\begin{equation}
\begin{aligned}
P_{k,\omega}&=\sum_{a\in\cA}\lambda_{k,\omega}^a\,P_{k,\omega}^a,\qquad
r_{k,\omega}=\sum_{a\in\cA}\lambda_{k,\omega}^a\,r_{k,\omega}^a,\\
c_{k,\omega}&=\sum_{a\in\cA}\lambda_{k,\omega}^a\,c_{k,\omega}^a.
\end{aligned}
\label{eq:node_aggregate}
\end{equation}

\end{proof}

\begin{remark}
\label{rem:no_tau_threshold}
In our case studies, the dynamics are player-separable: $x = [x^{(1)}, x^{(2)}]$, 
$A=\texttt{diag}(A_1,A_2)$, $B_1=[\tilde B_1;0]$, $B_2=[0;\tilde B_2]$, and the terminal quadratic has the signed block form
$Q_i=\texttt{diag}(Q_i^{(1)},-Q_i^{(2)})$ with $Q_i^{(1)},Q_i^{(2)}\succeq 0$.
The tree recursion preserves $P_{k,\omega}=\texttt{diag}(P^{(1)}_{k,\omega},-P^{(2)}_{k,\omega})$ with
$P^{(1)}_{k,\omega},P^{(2)}_{k,\omega}\succeq 0$, hence $H_{uv}=B_1^\top P^+B_2=0$ and
$H_{uu}=\tau\bar R + \tilde B_1^\top P^{(1)}\tilde B_1\succ0$, $H_{vv}=-\tau\bar S-\tilde B_2^\top P^{(2)}\tilde B_2\prec0$
for all $\tau>0$. Therefore every edge subproblem is strictly convex--concave and has a unique saddle without requiring $\tau<\tau^*$.
\end{remark}

\subsection{The Solver}
\label{subsec:node_aggregation}
For fixed $\alpha$, the tree Riccati recursion has two steps: The forward (Bayes) pass propagates beliefs $\{p_{k,\omega}\}$ and branch weights $\{\lambda_{k,\omega}^a\}$ using \eqref{eq:lambda}-\eqref{eq:bayes}, and compute averaged edge matrices $\{\bar R_{k,\omega}^a,\bar S_{k,\omega}^a\}$ via \eqref{eq:running_avg}. Then backward (Riccati) pass initializes leaf quadratics by \eqref{eq:leaf_params}; then for $k=K-1,\dots,0$ computes edge updates \eqref{eq:edge_feedback}-\eqref{eq:edge_c} and aggregate to nodes via \eqref{eq:node_aggregate}.





To enforce simplex constraints in \eqref{eq:alpha_def}, we parameterize $\alpha$ via logits $\phi$ and a softmax:
\begin{equation}
\alpha_{k,\omega,i}^a = \frac{\exp(\phi_{k,\omega,i}^a)}{\sum_{b\in\cA}\exp(\phi_{k,\omega,i}^b)}.
\label{eq:softmax}
\end{equation}
Alternatively, $\alpha$ can also be modeled by a neural network $F: \mathbb{R}^n \times \Delta (\Theta) \rightarrow \alpha$.
Given logits $\phi$, we compute $\alpha(\phi)$ and run the tree Riccati recursion to obtain the root quadratic $(P_{0,\emptyset},r_{0,\emptyset},c_{0,\emptyset})$. The scalar objective is then
\begin{equation}
\mathcal{L}(\phi):=V_{0,\emptyset}(x_0, p_0)
=\frac12 x_0^\top P_{0,\emptyset}x_0 + r_{0,\emptyset}^\top x_0 + c_{0,\emptyset}.
\label{eq:loss}
\end{equation}
All operations are differentiable, and thus gradients $\nabla_\phi \mathcal{L}$ can be obtained via automatic differentiation and used in first- or second-order optimizers. We can further accelerate the optimization by leveraging adjoint-enabled gradient descent \cite{amos2018differentiable}. The solver is summarized in Algorithm~\ref{alg:solver}.

\begin{algorithm}[!h]
\caption{LQ $G$ Solver}
\label{alg:solver}
\begin{algorithmic}[1]
\STATE \textbf{Input:} Game data $\{R_i,S_i,Q_i,q_i,c_i\}_{i=1}^I$, $(A,B_1,B_2)$, horizon $K$, time step $\tau$, initial state $x_0$, prior $p_0$.
\STATE Initialize logits $\{\phi_{k,\omega,i}^a\}$, choose optimizer and step size.
\FOR{outer iterations $t=0,1,2,\dots$}
    \STATE Compute $\alpha(\phi)$ via softmax \eqref{eq:softmax}.
    \STATE \textbf{(Forward)} Build belief tree $\{p_{k,\omega},\lambda_{k,\omega}^a\}$ \eqref{eq:lambda}--\eqref{eq:bayes}.
    \STATE \textbf{(Forward)} Compute $\bar R_{k,\omega}^a,\bar S_{k,\omega}^a$ \eqref{eq:running_avg}.
    \STATE \textbf{(Backward)} Initialize $(P_{K,\omega},r_{K,\omega},c_{K,\omega})$ at leaf \eqref{eq:leaf_params}.
    \FOR{$k=K-1,\dots,0$}
        \FOR{each node $\omega\in\cA^k$ and prototype $a\in\cA$}
            \STATE Compute edge feedback $(K_{k,\omega}^a,\kappa_{k,\omega}^a)$ via \eqref{eq:edge_feedback}.
            \STATE Compute $(P_{k,\omega}^a,r_{k,\omega}^a,c_{k,\omega}^a)$ via \eqref{eq:edge_P}--\eqref{eq:edge_c}.
        \ENDFOR
        \STATE Aggregate to node $(P_{k,\omega},r_{k,\omega},c_{k,\omega})$ via \eqref{eq:node_aggregate}.
    \ENDFOR
    \STATE Compute loss $\mathcal{L}(\phi)=V_{0,\emptyset}(x_0)$ via \eqref{eq:loss}.
    \STATE Backpropagate to obtain $\nabla_\phi \mathcal{L}$ from root to leaf.
    \STATE Optimizer step on $\phi$.
\ENDFOR
\STATE \textbf{Output:} Optimized $\alpha$ and edge feedback $\{K_{k,\omega}^a,\kappa_{k,\omega}^a\}$ to compose $\{\eta_i\}$.
\end{algorithmic}
\end{algorithm}


\vspace{-0.18in}
\section{Case Studies}
\vspace{-0.05in}
We implement the proposed method to solve two LQ $G$s. Both games use 2D double integrator dynamics and have a time horizon of $K=10$ with $\tau = 0.1$. Game-specific parameters are discussed as follows. 
\vspace{-0.08in}
\subsection{Case 1: Hexner's Game}
\vspace{-0.05in}
We use Hexner's game~\cite{hexner1979differential} with an analytical NE to verify the correctness of the learned policies. The game has a type-dependent terminal cost and type-independent running costs. P1 has two potential targets at $(0, \theta)$ with $\theta = \{-1, 1\}$. The terminal cost function is:
\begin{equation}
\label{eq:hexner_terminal}
\begin{aligned}
g_{\theta_i}(x) &= (x^{(1)} - z\theta_i)^\top \tilde{Q}_i(x^{(1)} - z\theta_i) - \\
& \hspace{1in}(x^{(2)} - z\theta_i)^\top \tilde{Q}_i(x^{(2)} - z\theta_i),
\end{aligned}
\end{equation}
where $x^{(1)}$ and $x^{(2)}$ represent the states of P1 and P2. $z = \begin{bmatrix}
    0 & 1 & 0 & 0
\end{bmatrix}^\top$. $\tilde{Q}_i = \tilde{Q} = \texttt{diag}(1, 1, 0, 0)$ is the terminal cost matrix. Rewriting \eqref{eq:hexner_terminal} in the form of \eqref{eq:terminal_cost_general}, we get,
\begin{align*}
&Q_i = Q = 2 \cdot \texttt{diag}(\tilde{Q}, -\tilde{Q}),\\
&q_i = 2\begin{bmatrix}-\theta_i (\tilde{Q}^\top z) & \theta_i (\tilde{Q}^\top z)\end{bmatrix}^\top,\\
&c_i = \theta_i^2 (z^\top \tilde{Q} z - z^\top \tilde{Q}z) = 0
\end{align*}
The control penalty matrices for P1 and P2 are $R = \texttt{diag}(0.05, 0.025)$ and $S = \texttt{diag}(0.05, 0.1)$, respectively. 

\begin{figure}[!htb]
    \centering
    \includegraphics[width=\linewidth]{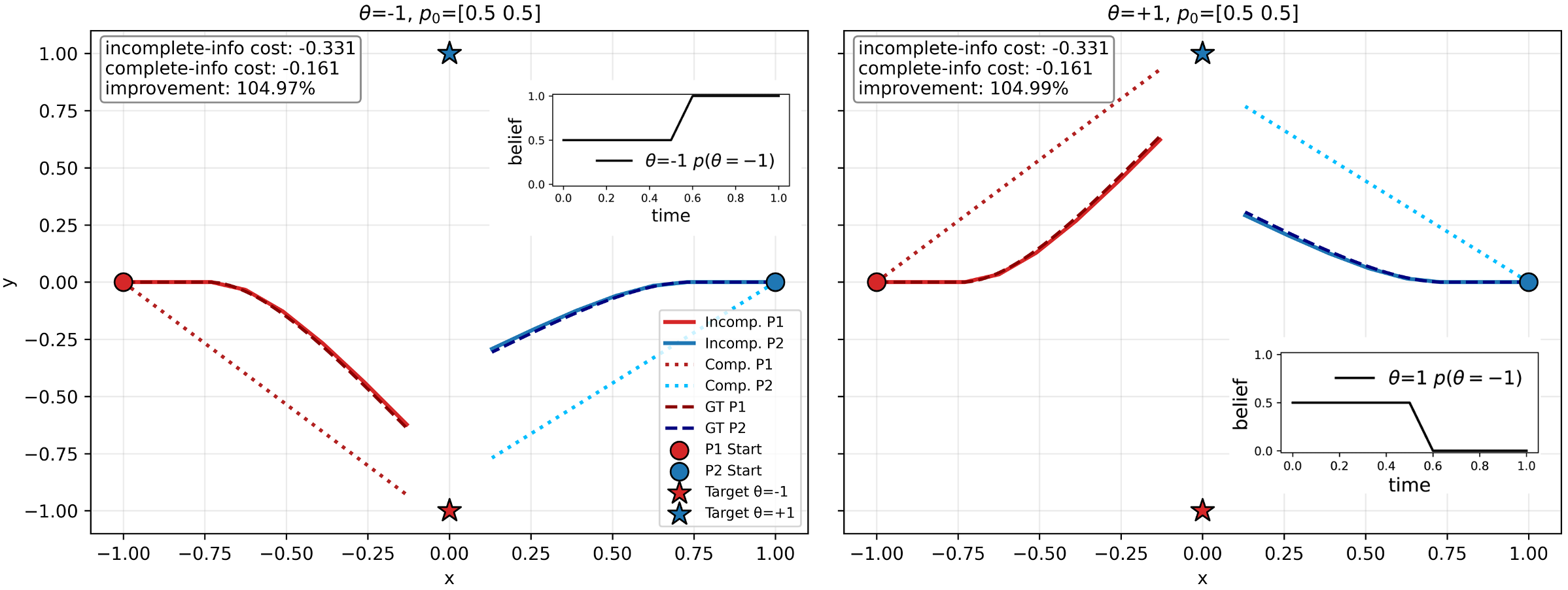}
    \vspace{-0.15in}
    \caption{Comparisons between ground-truth NEs (``GT'') of Hexner's game and learned policies for $\theta = -1$ (left) and $\theta = 1$ (right) under incomplete information. Dotted lines are NEs for the corresponding \textit{complete-information} games. P1 is able to improve significantly by exploiting its information advantage.}
    \label{fig:hexner_gt_comp}
    \vspace{-0.1in}
\end{figure}

The analytical solution for Hexner's game is studied in \cite{
ghimire2025solvingfootballexploitingequilibrium, hexner1979differential}: P1 plays non-revealing (i.e., its strategy is same across types) until a game-dependent critical time $t_r$ when they completely reveal the target. \cite{ghimire2025solvingfootballexploitingequilibrium} reported a cold-solve time of 5 minutes on a Macbook M1 Pro for P1's NE along with P2's best responses. In comparison, our algorithm achieves the same in under 1 minute. Results are visualized in Fig.~\ref{fig:hexner_gt_comp}, where the learned policies match the analytical ground-truth, recovering the critical time of $t_r \approx 0.5$. We also compare the NEs against those for the corresponding complete-information games (where P2 knows P1's targets) along with the value improvement percentages to show the necessity of information control in strategic planning.

\begin{figure*}[t]
\centering
\includegraphics[width=\linewidth]{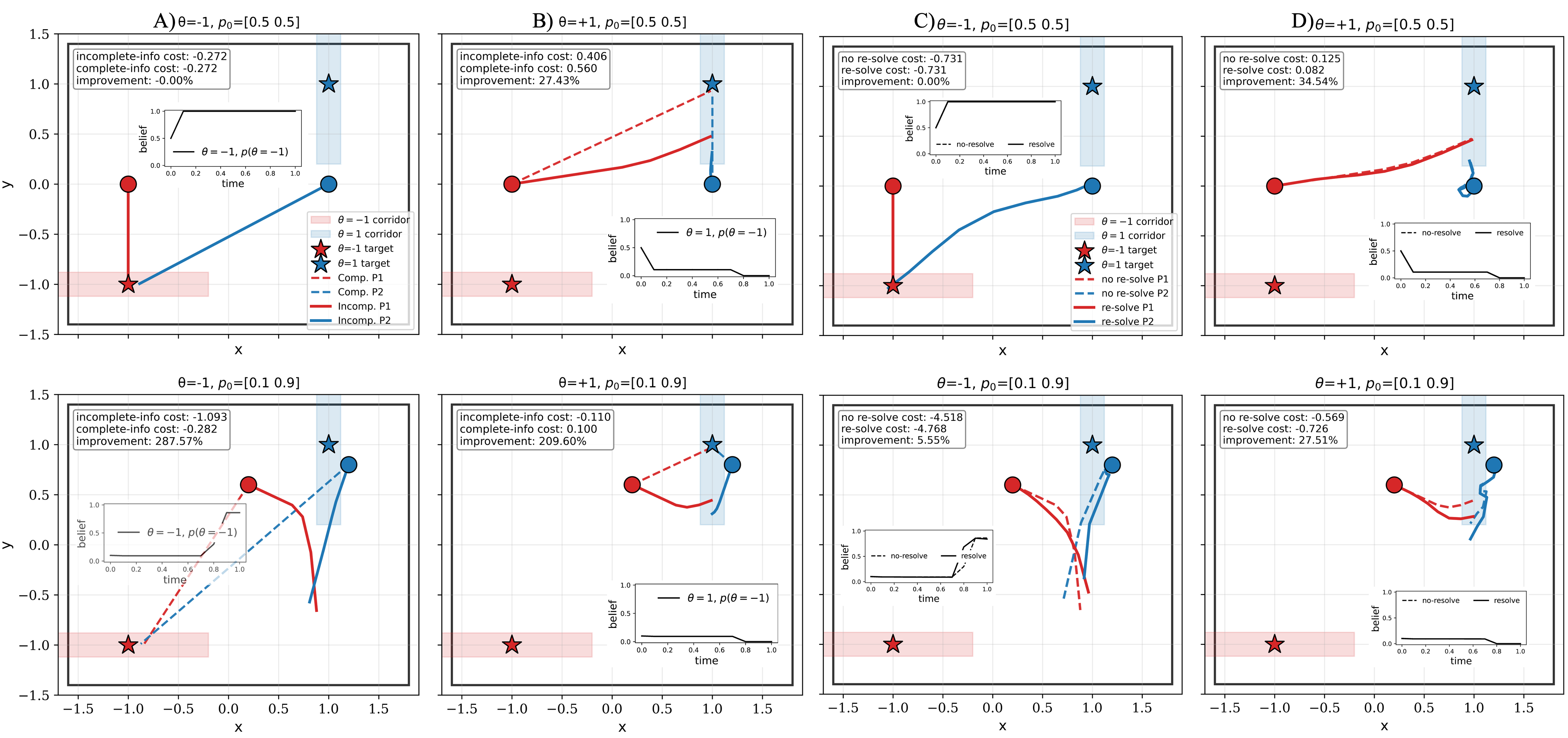}
\vspace{-0.25in}
\caption{Sample trajectories for the drone landing under attack case. A-B: Comparison between complete vs incomplete information games. P1 is better off in expectation when playing the incomplete-info game. Improvements for specific realizations are labeled in the plots. B-C: Comparison between re-solved vs offline policy when P2's dynamics is stochastic. On average, resolving helps P1.}
\label{fig:drone_case}
\vspace{-0.2in}
\end{figure*}
\vspace{-0.05in}
\subsection{Case 2: Drone Landing Under Attack}
\vspace{-0.05in}
In this game, a drone (P1) has two possible landing zones $(-1, -1)$ and $(1, 1)$ in a horizontal and a vertical orientation, and is evading from an adversary (P2). Final penalty matrices are type-dependent and shared by both players: $\tilde{Q}_1 = \texttt{diag}(1, 20, 0, 0)$ and $\tilde{Q}_2 = \texttt{diag}(20, 1, 0, 0)$, i.e., type 1 focuses on homing in $y$ coordinate and type 2 in $x$. Running penalty matrices are type-independent: $R = \texttt{diag}(0.05, 0.025)$, and $S = \texttt{diag}(0.02, 0.04)$. Dynamics, running and terminal cost functions follow Hexner's game, with $z = \begin{bmatrix}
    1 & 1 & 0 & 0
\end{bmatrix}^\top$.
\subsubsection{Optimal policy under nominal conditions}
First we present the result when the deployment and offline training share same condition. We highlight that, unlike in the Hexner's game, the optimal policy in this case is state- and belief-dependent. Fig.~\ref{fig:drone_case}A-B compare NE interactions for complete- and incomplete-info settings for different targets (columns) and different initial states (rows). 
\subsubsection{Optimal policy under stochastic disturbance}
Next, we introduce noise in P2's velocity during deployment, and argue that by re-solving the sub-game at each time-step, P1 can play more robustly. The noise is added from an i.i.d. Normal distribution $\mathcal{N}(0, \Sigma^2)$. To better reflect the underlying problem, P1 considers stochastic LQ during online play, which results in an accumulation of additional cost $\mathbf{tr}(P_k\Sigma)$ due to noise at every time step. However, note that the feedback gain remains unchanged in stochastic LQ with zero-mean noise. We use $\Sigma = \texttt{diag}(0, 0, 0.25, 0.25)$ in our results. 

We compare the cost of the game realized by re-solving the sub-games to that realized without re-solving. We report this difference in addition to the mean resolve time (milliseconds) per time step reported from 1600 runs in Tab.~\ref{tab:cost-table}. Sample trajectories are shown in Fig.~\ref{fig:drone_case}C-D.
\begin{table}[!h]
\vspace{-0.1in}
\centering
\caption{$\Delta$cost of P1 and mean subgame solving time}
\vspace{-0.15in}
\label{tab:cost-table}
\begin{tabular}{@{}ccc@{}}
\toprule
$\Delta$ cost $(\downarrow)$ &  $95\%$ Confidence Interval &resolve time (ms) \\ \midrule 
 $ -0.1051 \pm 0.3656$ & $[-0.1232,\; -0.0877]$ & $ 71.51  \pm 59.57 $ \\ \bottomrule
\end{tabular}
\end{table}
\vspace{-0.2in}
\section{Conclusion}
\vspace{-0.08in}
We provided an efficient gradient-based algorithm for solving the informed player's strategy in two-player zero-sum linear-quadratic differential games with one-sided information by leveraging its equilibrium structure. Future work will explore methods to improve algorithmic efficiency for the uninformed player, which solves a dual game that cannot exploit the LQ structure as much as the primal, yet may leverage the optimal signaling from the latter.
\bibliographystyle{ieeetr}
\bibliography{ref}

\begin{thebibliography}{10}

\bibitem{garcia2018design}
E.~Garcia, D.~W. Casbeer, and M.~Pachter, ``Design and analysis of state-feedback optimal strategies for the differential game of active defense,'' {\em IEEE Transactions on Automatic Control}, vol.~64, no.~2, pp.~553--568, 2018.

\bibitem{garcia2021complete}
E.~Garcia, D.~W. Casbeer, and M.~Pachter, ``The complete differential game of active target defense,'' {\em Journal of Optimization Theory and Applications}, vol.~191, no.~2, pp.~675--699, 2021.

\bibitem{liang2019differential}
L.~Liang, F.~Deng, Z.~Peng, X.~Li, and W.~Zha, ``A differential game for cooperative target defense,'' {\em Automatica}, vol.~102, pp.~58--71, 2019.

\bibitem{durkota2017optimal}
K.~Durkota, V.~Lis{\`y}, C.~Kiekintveld, K.~Hor{\'a}k, B.~Bo{\v{s}}ansk{\`y}, and T.~Pevn{\`y}, ``Optimal strategies for detecting data exfiltration by internal and external attackers,'' in {\em International Conference on Decision and Game Theory for Security}, pp.~171--192, Springer, 2017.

\bibitem{mc2016data}
S.~M. Mc~Carthy, A.~Sinha, M.~Tambe, and P.~Manadhata, ``Data exfiltration detection and prevention: Virtually distributed pomdps for practically safer networks,'' in {\em International Conference on Decision and Game Theory for Security}, pp.~39--61, Springer, 2016.

\bibitem{horak2016point}
K.~Hor{\'a}k and B.~Bo{\v{s}}ansk{\`y}, ``A point-based approximate algorithm for one-sided partially observable pursuit-evasion games,'' in {\em International conference on decision and game theory for security}, pp.~435--454, Springer, 2016.

\bibitem{van2013flipit}
M.~Van~Dijk, A.~Juels, A.~Oprea, and R.~L. Rivest, ``Flipit: The game of “stealthy takeover”,'' {\em Journal of Cryptology}, vol.~26, no.~4, pp.~655--713, 2013.

\bibitem{kyle1985continuous}
A.~S. Kyle, ``Continuous auctions and insider trading,'' {\em Econometrica: Journal of the Econometric Society}, pp.~1315--1335, 1985.

\bibitem{back1992insider}
K.~Back, ``Insider trading in continuous time,'' {\em The Review of Financial Studies}, vol.~5, no.~3, pp.~387--409, 1992.

\bibitem{wang2020market}
X.~Wang and M.~P. Wellman, ``Market manipulation: An adversarial learning framework for detection and evasion,'' in {\em 29th International Joint Conference on Artificial Intelligence}, 2020.

\bibitem{ghimire2025solvingfootballexploitingequilibrium}
M.~Ghimire, L.~Zhang, Z.~Xu, and Y.~Ren, ``Solving football by exploiting equilibrium structure of 2p0s differential games with one-sided information,'' 2025.

\bibitem{ghimirestate}
M.~Ghimire, L.~Zhang, Z.~Xu, and Y.~Ren, ``State-constrained zero-sum differential games with one-sided information,'' in {\em Forty-first International Conference on Machine Learning}.

\bibitem{harsanyi1967games}
J.~C. Harsanyi, ``Games with incomplete information played by “bayesian” players, i--iii part i. the basic model,'' {\em Management science}, vol.~14, no.~3, pp.~159--182, 1967.

\bibitem{aumann1995repeated}
R.~J. Aumann, M.~Maschler, and R.~E. Stearns, {\em Repeated games with incomplete information}.
\newblock MIT press, 1995.

\bibitem{cardaliaguet2007differential}
P.~Cardaliaguet, ``Differential games with asymmetric information,'' {\em SIAM journal on Control and Optimization}, vol.~46, no.~3, pp.~816--838, 2007.

\bibitem{cardaliaguet2009numerical}
P.~Cardaliaguet, ``Numerical approximation and optimal strategies for differential games with lack of information on one side,'' {\em Advances in Dynamic Games and Their Applications: Analytical and Numerical Developments}, pp.~1--18, 2009.

\bibitem{hexner1979differential}
G.~Hexner, ``A differential game of incomplete information,'' {\em Journal of Optimization Theory and Applications}, vol.~28, pp.~213--232, 1979.

\bibitem{elliott1972existence}
R.~J. Elliott and N.~J. Kalton, {\em The existence of value in differential games}, vol.~126.
\newblock American Mathematical Soc., 1972.

\bibitem{amos2018differentiable}
B.~Amos, I.~Jimenez, J.~Sacks, B.~Boots, and J.~Z. Kolter, ``Differentiable mpc for end-to-end planning and control,'' {\em Advances in neural information processing systems}, vol.~31, 2018.

\end{thebibliography}
\end{document}